\documentclass[%
 aip,
cp,  
amsmath,amssymb,
reprint,%
]{revtex4-2}
\usepackage[utf8]{inputenc}
\usepackage{amsthm,graphicx}
\usepackage{dcolumn}
\usepackage{bm,color}

\newtheorem{theorem}{{Theorem}}
\newtheorem{remark}{{Remark}}
\newtheorem{lemma}{Lemma}
\newtheorem{corrollary}{Corrollary}

\usepackage[utf8]{inputenc}
\usepackage[T1]{fontenc}
\usepackage{mathptmx}

\allowdisplaybreaks

\newcommand{\diag}{\mathrm{diag}}

\newcommand{\re}{\mathrm{Re}}
\newcommand{\im}{\mathrm{Im}}

\DeclareMathOperator{\dn}{\mathrm{dn}}
\DeclareMathOperator{\sn}{\mathrm{sn}}
\DeclareMathOperator{\cn}{\mathrm{cn}}

\def\openone{\leavevmode\hbox{\small1\kern-3.3pt\normalsize1}}
\def\im{\mbox{Im\,}}

\def\diag{\mbox{diag\,}}

\def\ad{\mbox{ad\,}}
\def\re{\mbox{Re\,}}

\def\m{{\mathbf{m}}}

\def\q{{\mathbf{q}}}

\def\p{{\mathbf{p}}}

\def\T{{\mathbf{T}}}

\def\q{{\mathbf{q}}}

\def\bbbr{{\Bbb R}}

\def\bbbc{{\Bbb C}}

\def\bbbz{{\Bbb Z}}
\def\bbbd{{\Bbb D}}

\def\ad{\mbox{ad}\,}
\def\diag{\mbox{diag}\,}

\def\re{{\rm Re}\,}
\def\im{{\rm Im}\,}

\def\fr#1{{\mathfrak{#1}}}
\def\openone{\leavevmode\hbox{\small1\kern-3.3pt\normalsize1}}

\begin{document}

\title{Fundamental Analytic Solutions for the Kulish-Sklyanin Model with Constant Boundary Conditions}

\author{$^{1, 2, 3}$Vladimir S. Gerdjikov} 
 \email[Corresponding author: ]{vgerdjikov@math.bas.bg}
\author{$^{1}$Aleksandr O. Smirnov}%
 \email{alsmir@guap.ru}
\affiliation{ 
$^{1}$Sankt-Petersburg State University of Aerospace
Instrumentation, St-Petersburg  B. Morskaya  67A St-Petersburg 190000 Russia \\[5pt]
$^{2}$Institute of Mathematics and Informatics Bulgarian Academy of Sciences \\
 8 Acad. G. Bonchev str. 1113 Sofia Bulgaria
 \\ [5pt]
$^{3}$Institute for Advanced Physical Studies, 111 Tsarigradsko chaussee,
Sofia 1784  Bulgaria
}


\date{\today} 

\begin{abstract}
In the present paper we analyze the construction of fundamental analytic solutions (FAS) for the generalized Kulish-Sklyanin models (KSM) for vanishing (VBC) and constant boundary conditions (CBC). Using FAS one can reduce the direct and inverse scattering problems for the Lax operator to a Riemann-Hilbert problem (RHP). For VBC we find two  FAS $\chi^+(x,t,\lambda) $ and $\chi^-(x,t,\lambda) $ analytic in the upper/lower $\mathbb{C}_\pm$ complex $\lambda$-plane. The RHP consists in: given the sewing function $G(x,t,\lambda)$ to constructing both $\chi^\pm(x,t,\lambda) $ in their regions of analyticity. For CBC the problem becomes more complicated, because now the RHP must be formulated on a Riemannian surface of genus 1.
\end{abstract}


\maketitle

\section{1.  Introduction}

In 1971 Zakharov and Shabat \cite{ZaSha01} discovered the integrability of the nonlinear Schr\"odinger equation (NLS). After the KdV equation, this was the second infinite-dimensional completely integrable system. In addition it found a number of applications in nonlinear optics, in plasma physics importance, hydrodynamics etc. About a year later the same authors discovered the integrability of the NLS equation under constant boundary conditions (CBC) \cite{ZaSha02}. The two equations look very similar:
\begin{equation}\label{eq:nls}\begin{aligned}
 \mbox{NLS$_0$:} \qquad &i \frac{\partial q_0}{ \partial x } + \frac{\partial^2 q_0 }{ \partial x^2 } + 2|q_0(x,t)|^2 q_0(x,t) =0; &\quad \mbox{NLS$_1$:} \qquad &i \frac{\partial q_1}{ \partial x } + \frac{\partial^2 q_1 }{ \partial x^2 } - 2(|q_1(x,t)|^2 -\rho^2) q_1(x,t) =0; \\
 & \lim_{x\to \pm \infty} q_0(x,t) = 0,  &\qquad  & \lim_{x\to \pm \infty} q_1(x,t) = \rho e^{i \theta_\pm},
\end{aligned}\end{equation}
 but in fact are substantially different. The reasons for this differences is not only in the sign of the nonlinearity. Let us first start with the similarities. Both equations allow Lax representations with very similar Lax operators $L_0$ and $L_1$ respectively:
 \begin{equation}\label{eq:}\begin{aligned}
 \mbox{NLS$_0$:} \quad & L_0 \colon i \frac{\partial \psi_0}{ \partial x }  + (Q_0(x,t) - \lambda \sigma_3 )\psi_0(x,t,\lambda)=0, & \quad  \mbox{NLS$_1$:} \quad & L_1 \colon i \frac{\partial \psi_1}{ \partial x }  + (Q_1(x,t) - \lambda \sigma_3 )\psi_1(x,t,\lambda)=0, \\
&  Q_0(x,t) = \left(\begin{array}{cc} 0 & q_0^* \\ q_0 & 0  \end{array}\right), &\quad
&  Q_1(x,t) = \left(\begin{array}{cc} 0 & q_1^* \\ -q_1 & 0  \end{array}\right).
 \end{aligned}\end{equation}
Both equations possess an infinite number of integrals of motion $C_{0;n}$ and $C_{1;n}$, soliton solutions, hierarchy of Hamiltonian formulations, fundamental analytic solutions (FAS) $\chi_0^\pm (x,\lambda)$ and $\chi_1^\pm (x,\lambda)$ etc.

Let us now briefly mention the differences. The Lax operator $L_0$ is not self-adjoint. Its continuous spectrum fills up the real axis of the complex $\lambda$-plane. It allows complex valued eigenvalues that go in pairs $\lambda_{0;j} = \mu_{0;j} \pm i \nu_{0;j}$. The soliton solution of NLS$_0$ is parametrized by 4 parameters: $\mu_{0j}$ - the velocity, $\nu_{0j}$ - the amplitude, and two additional constants that fix up the initial center of mass position $x_{0;j}$ and the initial phase $\phi_{0;j}$. The FAS $\chi_0^\pm (x,\lambda)$ are analytic functions of $\lambda$ for $\lambda \in \mathbb{C}_\pm$ respectively.

At the same time the Lax operator $L_1$ is  self-adjoint. In addition a new constant known as the `chemical potential` $\rho$ appears which is important for the relevant physical processes. Its continuous spectrum fills up two segments of real axis $(-\infty , -\rho] $ and  $[\rho, \infty ) $. It allows only real valued eigenvalues  $-\rho \leq \lambda_{1;j} \leq \rho$. The soliton solution of NLS$_1$ is parametrized by 2 parameters: the eigenvalue $\lambda_{1;j}$ - the velocity, and an additional constants that fixes up the initial center of mass position $x_{0;j}$. The most important differences come with the FAS: $\chi_1^\pm (x,\lambda)$ are analytic functions of $\lambda$ for $\lambda$ in the first and second leafs of the Riemannian surface related to $\sqrt{\lambda^2 -\rho^2}$ respectively, see \cite{ZaSha02, GeKu78, KaIno2, FaTa}.

Another set of important differences is related to their Hamiltonian structures. The phase space related to NLS$_0$ is linear. It is also well known that the inverse scattering transform for NLS$_0$ has the meaning of a generalized Fourier transform \cite{AKNS, DJK, GKh1, GKh2, GI-BJP10a, GI-BJP10b, GIK-TMF44, 152, GeYaV, KaupNewell, IP2, VSG2, ContMat}. For NLS$_1$ the phase space is a nonlinear one. In addition NLS$_1$ has a topological integral of motion $C_{1;0} = \theta_+ - \theta_-$; to each value of this integral there corresponds a sheaf of the phase space of NLS$_1$. In addition the CBC means that the integrals of motion (say, like the density of the particles or the Hamiltonian) need to be regularized. Additional important differences come up due to the fact that $\pm \rho$ are endpoints of the continuous spectrum of $L_1$ at which `virtual` eigenvalues may take place. Anyone can find detailed explanation  of these facts in the monograph by Faddeev and Takhtadjan \cite{FaTa}.

Our aim here will be to attempt an analysis of these differences for the multi-component NLS (MNLS) equations. It is well known that to each symmetric space \cite{Helg} one can relate an MNLS \cite{ForKu*83}. The famous Manakov model \cite{ma74} in which $q_0$ becomes a two-component vector is related to the space $SU(3)/S(U(2)\times U(1))$. In what follows we will fix up our main attention to the Kulish-Sklyanin model (KSM) which is related to a symmetric space of BD.I-type: $SO(2n+1)/SO(2n-1)\times O(2))$; KSM is obtained with $n=3$. Some of these MNLS with CBC have already been studied, see \cite{IROSI, Bersano, CYHo, PriViBio, AblPrinTru*04, Lanning, PriAbBio, PP1, LiBP, Tsu1}. Unfortunately in most of them the authors have imposed constraints on the asymptotic values of the potential $Q_\pm$ that greatly simplify the spectral analysis of the corresponding Lax operator. Below the only condition that we will impose on the asymptotic operators $L_\pm$ will be that they have the same spectrum. Such condition are compatible with the integrability of the corresponding equation. We will also limit ourselves with the KSM because it is important from the point of view of applications (describes spin-1 Bose-Einstein condensate (BEC)) and because it leads to nontrivial results.

In fact such analysis has already been done in a slightly different framework, see \cite{PriViBio, PriAbBio, PP1, LiBP} closer to the one found to describe spin-1 BEC in \cite{IMW04, IMM07, LLMML05, OM, UIW}. There the Lax representation is given by $4\times 4$ matrices which belong to the algebra $sp(4)$ and the corresponding model is termed as the matrix NLS:
\begin{equation}\label{eq:tL}\begin{split}
 \tilde{L} \tilde{\psi} &\equiv i \frac{\partial \tilde{\psi} }{ \partial x } + (\tilde{Q}(x,t) - \lambda \tilde{J} )\tilde{\psi}(x,t,\lambda) =0, \\
 \tilde{Q} &= \left(\begin{array}{cc} 0 & \tilde{\q} \\ \tilde{\p} & 0   \end{array}\right), \qquad \tilde{J} = \left(\begin{array}{cc} \openone_2 & 0 \\ 0 & -\openone_2   \end{array}\right), \qquad \tilde{\q} = \left(\begin{array}{cc} q_0 & q_1 \\ q_{-1} & q_0   \end{array}\right), \qquad
  \tilde{\p} = - \tilde{\q}^\dag.
\end{split}\end{equation}
Note that the algebra $sp(4)$ is isomorphic $so(5)$ and in fact the Lax operator $\tilde{L}$ is the same operator $L$ written in the spinor representation of $so(5)$. The authors of \cite{AblPrinTru*04, PriViBio, PriAbBio, PP1, LiBP} have considered a special class of CBC constrained by the conditions:
\begin{equation}\label{eq:tQpm}\begin{split}
 \lim_{x\to\pm\infty} \tilde{Q} = \tilde{Q}_\pm, \qquad \tilde{Q}_\pm^2 = k_0^2 \openone_2.
\end{split}\end{equation}
This constraint, on one side simplifies substantially the construction of the FAS. On the other hand may have put limits on the physical applications, because it may be rather difficult to ensure that the experimental values of $\tilde{Q}_\pm$ satisfy (\ref{eq:tQpm}).

The paper is organized as follows. In the next Section 2 we outline the main facts about the theory of generalized KSM related to any of the symmetric spaces $SO(2n+1)/SO(2n-1)\times O(2))$ for VBC. In the next Section 3 we analyze the effect of the CBC on the Lax operator.
We derive the characteristic polynomials for $U_\pm$ for generic CBC. These polynomials always have a vanishing eigenvalues $z_0=0$. We consider the nontrivial part of this polynomial and find that it defines a genus 1 curve which we parametrize in terms of the uniformizing variable $u$ and the Jacobi elliptic functions. Next  we construct the FAS of $L$. To this end we need to find the curves in the fundamental domain on which either $\im z_j(u) =0$, or $\im (z_j(u)\pm z_k(u)) =0$, $k\neq j$. These curves define the continuous spectrum of $L_\pm$ and $L$. One can also reduce the ISP for $L$ to a RHP on these curves. Section 4 is for discussions and conclusions.

\section{2. Kulish-Sklyanin model for vanishing boundary conditions}
We start with the Lax representation for the generalized Kulish-Sklyanin model (KSM). We first fix up the class of Lie algebras $\mathfrak{g} \simeq B_r$ and consider the generalized Zakharov-Shabat systems \cite{ZMNP, ZaSha01} which are of the form:
 \begin{equation}\label{eq:L}\begin{split}
 L\psi \equiv i \frac{\partial \psi}{ \partial x } + (Q(x,t) - \lambda J)\psi(x,t,\lambda) =0.
 \end{split}\end{equation}
 For KSM we fix up $J$ to be dual to the basis vector $\vec{e}_1$ in the root system. These choices determine the phase space $\mathcal{M}$ of the relevant NLEE or, in other words, the space of allowed potentials is defined as:
\begin{equation}\label{eq:M}\begin{split}
 \mathcal{M} \equiv \left\{ Q(x,t) = [J,X(x,t)], \qquad X(x,t) \in \mathfrak{g} \right\},
\end{split}\end{equation}
i.e. $Q(x,t)$ belongs to the co-adjoint orbit of $\mathcal{O}_J \in \mathfrak{g}$ passing through $J$. The second Lax operator in the pair for KSM is:
\begin{equation}\label{eq:UV1}\begin{aligned}
M\psi \equiv i \frac{\partial \psi}{ \partial t } + V^{(1)}(x,t,\lambda)\psi (x,t,\lambda) &= 0, \qquad V^{(1)}(x,t,\lambda) = i Q_{1,x}+ V_2(x,t) +\lambda V_1(x,t) -\lambda^2 J,
\end{aligned}\end{equation}
where $Q_1(x,t) =\ad_J Q(x,t)$, $V_1(x,t) =Q(x,t)$ and
\begin{equation}\label{eq:UV4}\begin{split}
Q(x,t) = \left(\begin{array}{ccc} 0 & \vec{q}^T & 0 \\ \vec{p} & 0 & s_0 \vec{q} \\
0 & \vec{p}^Ts_0 & 0 \end{array}\right), \qquad   V_2(x,t) = \frac{1}{2}\ad_{Q_1} Q(x,t) = \left(\begin{array}{ccc} (\vec{q},\vec{p}) & 0 & 0 \\
0 & s_0 \vec{q} \vec{p}^T s_0 - \vec{p}\vec{q}^T & 0 \\ 0 &  0 & -(\vec{q},\vec{p}) \end{array}\right).
\end{split}\end{equation}
In order to get the generalized KSM as the compatibility condition on this Lax pair we have to impose the constraint $Q(x,t) = Q^\dag (x,t)$, or $\vec{p}(x,t) =\vec{q}^*(x,t)$:
\begin{equation}\label{eq:KSm}\begin{split}
i \frac{\partial \vec{q}}{ \partial t} + \frac{\partial^2 \vec{q} }{ \partial x^2 } + 2 (\vec{q}\;^\dag, \vec{q})\vec{q}
- (\vec{q}^T s_0 \vec{q}) s_0 \vec{q}\;^* =0, \quad s_0 = \sum_{k=1}^{2r-1} (-1)^k E_{k,2r-k}.
\end{split}\end{equation}
where now $E_{kn}$ is a $2r-1 \times 2r-1$-matrix with $(E_{kn})_{pj}=\delta_{kp}\delta_{nj}$.
The well known Kulish-Sklyanin model whose integrability has been known since 1981
\cite{KuSkl} corresponds to $\mathfrak{g}\simeq so(5)$; then the vectors $\vec{p} =\vec{q}^*$ are three-component.

For applications of this model to Bose-Einstein condensates and detailed analysis for the inverse spectral transform see \cite{IMW04, IMM07, LLMML05, OM, UIW,  VSG2, PriViBio, Ho}.

The spectral properties of Lax operators and the relevant FAS are well known, see \cite{VSG2, GGK-TMF144, GGK05a, 1, GGK05b, VRG-WMo}. Here we remind the main facts about them.

We will use  the Jost solutions  which are defined by, see \cite{VSG2} and the references therein
\begin{equation}
\lim_{x \to -\infty} \phi(x,t,\lambda) e^{  i \lambda^k J x }=\openone, \qquad  \lim_{x \to \infty}\psi(x,t,\lambda) e^{  i \lambda^k J x } = \openone
 \end{equation}
and the scattering matrix $T(\lambda,t)\equiv \psi^{-1}\phi(x,t,\lambda)$. Due to the special choice of $J$ and
to the fact that the Jost solutions and the scattering matrix take values in the group $SO(2r+1)$ we can use the following
block-matrix structure of $T(\lambda,t)$ and its inverse $\hat{T}(\lambda,t)$:
\begin{equation}\label{eq:25.1}
T(\lambda,t) = \left( \begin{array}{ccc} m_1^+ & -\vec{B}^-{}^T & c_1^- \\
\vec{b}^+ & {\bf T}_{22} & -s_0\vec{b}^- \\ c_1^+ & \vec{B}^+{}^Ts_0 & m_1^- \\ \end{array}\right), \qquad
\hat{T}(\lambda,t) = \left( \begin{array}{ccc} m_1^- & \vec{b}^-{}^T & c_1^- \\
-\vec{B}^+ & \hat{\bf T}_{22} & s_0\vec{B}^- \\ c_1^+ & -\vec{b}^+{}^Ts_0 & m_1^+ \\ \end{array}\right),
\end{equation}
where $\vec{b}^\pm (\lambda,t)$ and $\vec{B}^\pm (\lambda,t)$ are $2r-1$-component vectors,
${\bf T}_{22}(\lambda)$ and $\hat{\bf T}_{22}(\lambda)$ are  $2r-1 \times 2r-1$ blocks and $m_1^\pm (\lambda)$,  $c_1^\pm (\lambda)$ are scalar functions satisfying
$c_1^\pm = 1/2 (\vec{b}^\pm \cdot s_0 \vec{b}^\pm) /m_1^\mp$.

Let us now introduce $\Phi(x,t,\lambda) =\phi(x,t,\lambda) e^{i \lambda J x }$ and  $\Psi(x,t,\lambda) =\psi(x,t,\lambda) e^{i \lambda J x }$.  Then one can check that
$\Psi(x,t,\lambda)$ and $\Phi(x,t,\lambda)$ must satisfy the following Volterra-type equations:
\begin{equation}\label{eq:Psi2}\begin{split}
\Psi(x,t,\lambda) & = \openone + i \int_{\infty}^{x} dy\; e^{-i\lambda J(x-y)} Q(y,t) \Psi(y,t,\lambda) e^{i\lambda J(x-y)}, \\
\Phi(x,t,\lambda) & = \openone + i \int_{-\infty}^{x} dy\; e^{-i\lambda J(x-y)} Q(y,t) \Phi(y,t,\lambda) e^{i\lambda J(x-y)}.
\end{split}\end{equation}
Note that  from the equations (\ref{eq:Psi2}) there follows that the first column $\Phi(x,t,\lambda) $ and the last column of $\Psi(x,t,\lambda) $ are analytic functions of $\lambda$ for $\im \lambda >0$. Likewise the first column $\Psi(x,t,\lambda) $ and the last column of $\Phi(x,t,\lambda) $ are analytic functions of $\lambda$ for $\im \lambda <0$. The middle columns of both $\Psi(x,t,\lambda) $ and $\Phi(x,t,\lambda) $ are defined only on the real $\lambda$ axis.

However, following the ideas of Shabat \cite{Sh*75, Sh*79, ZMNP} we can introduce the integral equations
\begin{equation}\label{eq:Xi2}\begin{split}
\xi^+_{jk}(x,t,\lambda) & = \delta_{jk} + i \int_{\infty}^{x} dy\; (Q(y,t) \xi^+(y,t,\lambda))_{jk}   e^{-i\lambda (J_j - J_k)(x-y)} , \qquad \mbox{for} \quad j <k \\
\xi_{jk}^+(x,t,\lambda) & = \delta_{jk} + i\int_{-\infty}^{x} dy\; (Q(y,t) \xi^+(y,t,\lambda))_{jk} e^{-i\lambda (J_j - J_k)(x-y)} , \qquad \mbox{for} \quad j \geq k.
\end{split}\end{equation}
\begin{equation}\label{eq:Xi20}\begin{split}
\xi^{\prime, +}_{jk}(x,t,\lambda) & = \delta_{jk} + i \int_{\infty}^{x} dy\; (Q(y,t) \xi^{\prime, +}(y,t,\lambda))_{jk}   e^{-i\lambda (J_j - J_k)(x-y)} , \qquad \mbox{for} \quad j \leq k \\
\xi_{jk}^{\prime, +}(x,t,\lambda) & = \delta_{jk} + i\int_{-\infty}^{x} dy\; (Q(y,t) \xi^{\prime, +}(y,t,\lambda))_{jk} e^{-i\lambda (J_j - J_k)(x-y)} , \qquad \mbox{for} \quad j > k.
\end{split}\end{equation}
Both sets of solution $\xi^+(x,t,\lambda)$ and $\xi^{\prime, +}(x,t,\lambda)$ of eqs. (\ref{eq:Xi2}) and (\ref{eq:Xi20}) will be an analytic function for $\im \lambda >0$.

Similarly we can introduce $\xi^-(x,t,\lambda)$ as the solution of the following integral equations:
\begin{equation}\label{eq:Xi3}\begin{split}
\xi^-_{jk}(x,t,\lambda) & = \delta_{jk} + i \int_{-\infty}^{x} dy\; (Q(y,t) \xi^-(y,t,\lambda))_{jk}   e^{-i\lambda (J_j - J_k)(x-y)} , \qquad \mbox{for} \quad j < k \\
\xi_{jk}^-(x,t,\lambda) & = \delta_{jk} + i\int_{\infty}^{x} dy\; (Q(y,t) \xi^-(y,t,\lambda))_{jk} e^{-i\lambda (J_j - J_k)(x-y)} , \qquad \mbox{for} \quad j \geq k.
\end{split}\end{equation}
\begin{equation}\label{eq:Xi30}\begin{split}
\xi^{\prime, -}_{jk}(x,t,\lambda) & = \delta_{jk} + i \int_{-\infty}^{x} dy\; (Q(y,t) \xi^{\prime, -}(y,t,\lambda))_{jk}   e^{-i\lambda (J_j - J_k)(x-y)} , \qquad \mbox{for} \quad j \leq k \\
\xi_{jk}^{\prime, -}(x,t,\lambda) & = \delta_{jk} + i\int_{\infty}^{x} dy\; (Q(y,t) \xi^{\prime, -}(y,t,\lambda))_{jk} e^{-i\lambda (J_j - J_k)(x-y)} , \qquad \mbox{for} \quad j > k.
\end{split}\end{equation}
Both sets of solution $\xi^-(x,t,\lambda)$ and $\xi^{\prime, -}(x,t,\lambda)$ of eqs. (\ref{eq:Xi3}) and (\ref{eq:Xi30}) will be an analytic function for $\im \lambda <0$.

Now we need to recall that any two fundamental solutions of the operator $L$ must be linearly dependent. In particular $\xi^\pm(x,t,\lambda)$ and $\xi^{\prime, \pm}(x,t,\lambda)$ could b expressed through the Jost solutions. The corresponding coefficients can be found evaluating the asymptotics of $\xi^\pm (x,t,\lambda)$ and $\xi^{\prime, \pm}(x,t,\lambda)$ for $x \to \infty$ and $x \to -\infty$. The results are:
\begin{equation}\label{eq:xipm}\begin{aligned}
 \xi^\pm (x,t,\lambda) &= \phi (x,t,\lambda) S^\pm_J(t,\lambda), &\qquad  \xi^\pm (x,t,\lambda) &= \psi (x,t,\lambda) T^\mp_J(t,\lambda) D^\pm_J(\lambda), \\
 \xi^{\prime, \pm} (x,t,\lambda) &= \phi (x,t,\lambda) S^\pm_J(t,\lambda) \hat{D}^\pm_J(\lambda), &\qquad  \xi^{\prime, \pm} (x,t,\lambda) &= \psi (x,t,\lambda) T^\mp_J(t,\lambda) ,
\end{aligned}\end{equation}
where the factors $T_J^\pm (t,\lambda)$, $D_J^\pm (\lambda)$ and $S_J^\pm (t,\lambda)$  are related to the scattering matrix as follows:
\begin{equation}\label{eq:TGauss}\begin{split}
 T(t,\lambda) = T_J^- (t,\lambda) D_J^+ (\lambda) \hat{S}_J^+ (t,\lambda) =  T_J^+ (t,\lambda) D_J^- (\lambda) \hat{S}_J^- (t,\lambda).
\end{split}\end{equation}
In other words  $S_J^\pm (t,\lambda)$, $T_J^\pm (t,\lambda)$ are the factors in the Gauss decomposition of $T(t,\lambda)$. They can be expressed through the matrix elements of $T(t,\lambda)$ as follows:

\begin{description}

  \item[i)] the factors $S_J^\pm (t,\lambda)$, $T_J^\pm (t,\lambda)$  are related to the scattering matrix as follows:
  \begin{equation}\label{eq:SJpm}\begin{aligned}
  S_J^\pm (t,\lambda) & = \exp \left(\pm \sum_{\alpha \in \Delta_1^+}^{} \tau^\pm_\alpha E_{\pm \alpha} \right), \qquad   T_J^\pm (t,\lambda) & = \exp \left(\pm \sum_{\alpha \in \Delta_1^+}^{} \rho^\mp_\alpha E_{\pm \alpha} \right).
  \end{aligned}\end{equation}
where
 \begin{equation}\label{eq:rho-tau}\begin{split}
  \vec{\rho}^\pm(t,\lambda) = \frac{ \vec{b}^\pm}{m_1^\pm}, \qquad \vec{\tau}^\pm(t,\lambda) =  \frac{ \vec{B}^\pm}{m_1^\pm},
 \end{split}\end{equation}
see Appendix C. 

  \item[ii)] $D_J^\pm (\lambda)$ are block-diagonal functions
  \begin{equation}\label{eq:Dpm}\begin{split}
D_J^+ &= \left( \begin{array}{ccc} m_1^+ & 0 & 0 \\ 0 & {\bf m}_2^+ & 0 \\
0 & 0 & 1/m_1^+ \end{array} \right), \qquad  D_J^- = \left( \begin{array}{ccc} 1/m_1^- & 0 & 0 \\ 0 & {\bf m}_2^- & 0 \\
0 & 0 & m_1^- \end{array} \right),
  \end{split}\end{equation}
  analytic for $\lambda \in \mathbb{C}_\pm$ respectively; in addition
\begin{equation}\label{eq:rotau}\begin{aligned}
\m_2^+ &= \T_{22}+ \frac{\vec{b}^+ \vec{b}^{-,T} }{2m_1^+}  =\hat{ \T}_{22}+ \frac{s_0\vec{b}^- \vec{b}^{+,T}s_0 }{2m_1^+} , \qquad \m_2^- &= \hat{\T}_{22}+ \frac{\vec{B}^+ \vec{B}^{-,T} }{2m_1^-}= \hat{\T}_{22}+ \frac{s_0\vec{B}^- \vec{B}^{+,T}s_0 }{2m_1^-} .
\end{aligned}\end{equation}

\end{description}


\subsection{2.1 The direct and inverse scattering problems for $L$}

Next we remind that solving the  direct and the inverse scattering problem (ISP) for $L$
is reduced to  a Riemann-Hilbert problem (RHP) for the fundamental analytic solution (FAS) $\chi^{\pm} (x,t,\lambda )$. Their construction is based on the generalized Gauss decomposition of $T(\lambda,t)$
\begin{equation}\label{eq:FAS_J}
\chi ^\pm(x,t,\lambda)= \phi (x,t,\lambda) S_{J}^{\pm}(t,\lambda ) = \psi (x,t,\lambda ) T_{J}^{\mp}(t,\lambda ) D_J^\pm (\lambda).
\end{equation}
Here $S_{J}^{\pm} $ and $T_{J}^{\pm} $ are upper- and lower-block-triangular matrices, while $D_J^\pm(\lambda)$ are block-diagonal matrices with the same block structure as $T(\lambda,t)$ above. The explicit expressions of these Gauss factors in terms of the matrix elements of $T(\lambda,t)$ are given above by equations (\ref{eq:SJpm}) and (\ref{eq:rho-tau}).

If $\vec{q}(x,t) $ evolves according to (\ref{eq:KSm}) then the scattering matrix and its elements satisfy the following linear evolution equations
\begin{equation}\label{eq:evol}
i\frac{d\vec{B}^{\pm}}{d t} \pm \lambda ^{2k} \vec{B}^{\pm}(t,\lambda ) =0,  \qquad
i\frac{d\vec{b}^{\pm}}{d t} \pm \lambda ^{2k} \vec{b}^{\pm}(t,\lambda ) =0,  \qquad  i\frac{d m_1^{\pm}}{d t}  =0,
 \qquad  i \frac{d{\bf m}_2^{\pm}}{d t}  =0,
\end{equation}
so the block-diagonal matrices $D^{\pm}(\lambda)$ are  generating functionals of the integrals of motion.
The fact that all $(2r-1)^2$ matrix elements of $m_2^\pm(\lambda)$ for $\lambda \in \bbbc_\pm$  generate integrals of motion reflects
the super-integrability of the model and is due to the degeneracy of the dispersion law determined by $\lambda^{2k} J$. We remind that
$D^\pm_J(\lambda)$ allow analytic extension for $\lambda\in \bbbc_\pm$ and that their zeroes and
 poles determine the discrete eigenvalues of $L$. We will use also another set of FAS:
\begin{equation}\label{eq:chi'}\begin{split}
\chi^{\prime,\pm} (x,t,\lambda) = \chi^\pm (x,t,\lambda) \hat{D}^\pm_J (\lambda).
\end{split}\end{equation}

The FAS for real $\lambda$ are linearly related
\begin{equation}\label{eq:rhp0}\begin{split}
\chi^+(x,t,\lambda) &=\chi^-(x,t,\lambda) G_J(\lambda,t), \qquad G_{0,J}(\lambda,t) =S^-_J(\lambda,t)S^+_J(\lambda,t) , \\
\chi^{\prime,+}(x,t,\lambda) &=\chi^{\prime,-}(x,t,\lambda) G'_J(\lambda,t), \qquad G'_{0,J}(\lambda,t) =T^+_J(\lambda,t)T^-_J(\lambda,t) .
\end{split}\end{equation}
One can rewrite eq. (\ref{eq:rhp0}) in an equivalent form for the FAS
$\xi^\pm(x,t,\lambda)=\chi^\pm (x,t,\lambda)e^{i\lambda Jx }$ and $\xi^{\prime,\pm}(x,t,\lambda)=\chi^{\prime,\pm} (x,t,\lambda)e^{i\lambda Jx }$
which satisfy also the relation
\begin{equation}\label{eq:rh-n}
\lim_{\lambda \to \infty} \xi^\pm(x,t,\lambda) = \openone, \qquad \lim_{\lambda \to \infty} \xi^{\prime,\pm}(x,t,\lambda) = \openone.
\end{equation}
Then for $\im \lambda=0$ these FAS satisfy
\begin{equation}\label{eq:rhp1}\begin{aligned}
\xi^+(x,t,\lambda) &=\xi^-(x,t,\lambda) G_J(x,\lambda,t), &\quad G_{J}(x,\lambda,t) &=e^{-i\lambda
Jx -i \lambda^2 Jt}G_{0,J}(\lambda,t)e^{i\lambda Jx +i \lambda^2 Jt} ,  \\
\xi^{\prime,+}(x,t,\lambda) &=\xi^{\prime, -}(x,t,\lambda) G'_J(x,\lambda,t), &\quad G'_{J}(x,\lambda,t) &=e^{-i\lambda Jx -i \lambda^2 Jt}G'_{0,J}(\lambda,t)e^{i\lambda Jx +i \lambda^2 Jt} .
\end{aligned}\end{equation}
Obviously the sewing function $G_j(x,\lambda,t)$ is uniquely determined by
the Gauss factors of $T(\lambda,t)$. In view of eq. (\ref{eq:SJpm}) we arrive to the following
\begin{lemma}\label{lem:ms}
Let the potential $Q(x,t)$ be such that the Lax operator $L$ has no discrete eigenvalues. Then as minimal set of scattering data
which determines uniquely the scattering matrix $T(\lambda,t)$ and the corresponding potential $Q(x,t)$ one can consider either one
of the sets $\mathfrak{T}_i$, $i=1,2$
\begin{equation}\label{eq:T_i}
\mathfrak{T}_1 \equiv \{ \vec{\rho}^+(\lambda,t), \vec{\rho}^-(\lambda,t),  \quad \lambda^k \in \bbbr\},
\qquad \mathfrak{T}_2 \equiv \{ \vec{\tau}^+(\lambda,t), \vec{\tau}^-(\lambda,t),  \quad \lambda^k \in \bbbr\}.
\end{equation}\end{lemma}

\begin{proof} i) From the fact that $T(\lambda,t)\in SO(2r+1)$ one can derive that
\begin{equation}\label{eq:25.3}
\frac{1}{m_1^+m_1^-} = 1 + (\vec{\rho^+},\vec{\rho^-}) + \frac{1}{4}
(\vec{\rho^+},s_0\vec{\rho^+}) (\vec{\rho^-}, s_0\vec{\rho^-})
\end{equation}
for $\lambda\in\bbbr$. Using the analyticity properties of $m_1^\pm $  we can recover them from eq. (\ref{eq:25.3})
using Cauchy-Plemelji formulae. Given $\mathfrak{T}_i$ and $m_1^\pm$ one easily recovers $\vec{b}^\pm(\lambda)$
and $c_1^\pm(\lambda)$. In order to recover ${\bf m}_2^\pm$ one again uses their analyticity properties,
 only now the problem reduces to a RHP for functions on $SO(2r+1)$. The details will be presented elsewhere.

ii) Given $\mathfrak{T}_i$ one uniquely recovers the sewing function $G_J(x,t,\lambda)$.
In order to recover the corresponding potential $Q(x,t)$ one can use the fact that the RHP (\ref{eq:rhp1})
with canonical normalization has unique regular solution $\chi^\pm(x,t,\lambda)$. Given  $\chi^\pm(x,t,\lambda)$ we recovers $Q(x,t)$ via:
\begin{equation}\label{eq:QQ}
 Q(x,t) = \lim_{\lambda\to\infty} \lambda \left( J - \chi^\pm J \widehat{\chi}^\pm(x,t,\lambda)\right) =
 \lim_{\lambda\to\infty} \lambda \left( J - \chi^{\prime,\pm} J \widehat{\chi}^{\prime,\pm}(x,t,\lambda)\right)..
\end{equation}
which is well known.
\end{proof}
As a consequence of the standard reduction  $Q(x,t)=Q^\dag(x,t)$ (or in components $p_k=q_k^*$) we have $\vec{\rho}^-(\lambda,t)=\vec{\rho}^{+,*}(\lambda,t)$ and $\vec{\tau}^-(\lambda,t)=\vec{\tau}^{+,*}(\lambda,t)$.

For the sake of completeness we also write down the Hamiltonian of the generalized KSM (\ref{eq:KSm}) for vanishing boundary conditions:
\begin{equation}\label{eq:HKSm0}\begin{split}
H_{\rm VBC} = \int_{-\infty}^{\infty} dx\; \left( \left( \frac{\partial \vec{q}^\dag }{ \partial x}, \frac{\partial \vec{q} }{ \partial x } \right) - (\vec{q}^\dag, \vec{q})^2 + \frac{1}{2} (\vec{q}^T
s_0 \vec{q}) (\vec{q}^\dag s_0 \vec{q}^*) \right).
\end{split}\end{equation}

\section{3. Kulish-Sklyanin model for constant boundary conditions}

Our aim will be to analyze the properties of the KSM for constant boundary conditions (CBC). Typically CBC for NLS type equations are related with substantial changes on the for and spectrum of the Lax pair. The first of these changes is the reduction on the potential $Q(x,t)$ which now reads $Q^\dag (x,t) = - Q(x,t)$, i.e. $\vec{p} = - \vec{q}^*$. Another obvious fact is that we can not use as Hamiltonian (\ref{eq:HKSm0}) even if we take into account the new reduction and change the signs correspondingly. Indeed, let the vector $\vec{q}(x,t)$ and the potential $Q(x,t)$ for $x\to\pm\infty$ tend fast enough to
\begin{equation}\label{eq:Qpm}\begin{split}
 \lim_{x\to\pm\infty} Q(x,t) = Q_\pm , \qquad Q_\pm = \left(\begin{array}{ccc} 0 & \vec{q}_\pm^T & 0 \\ \vec{p}_\pm & 0 & s_0 \vec{q}_\pm \\ 0 & \vec{p}_\pm^Ts_0 & 0 \end{array}\right).
\end{split}\end{equation}
Then the integrand in (\ref{eq:HKSm0}) will also tend to a constant for $x\to \pm \infty$ and the integral will be divergent. Therefore the Hamiltonian, the corresponding KSM and also the $M$ operator will have to be modified in order to have coherent theory.

We will need also:
\begin{equation}\label{eq:Upm}\begin{split}
 U_\pm (\lambda) = Q_\pm - \lambda J, \qquad L_\pm \psi_\pm = i \frac{\partial \psi_\pm}{ \partial x } + U_\pm (\lambda) \psi_\pm (x,t,\lambda) =0.
\end{split}\end{equation}
We remind also that we will consider mainly the KSM, i.e. the case of 3-component vectors $\vec{q}$ and $\vec{p}$. From physical point of view it describes the spin-1 BEC. The cases when these vectors have 5 (spin-2 BEC) or more component lead to similar results for the spectrum of $L_\pm$.
It is well known that the integrability properties are preserved provided the asymptotic operators $L_\pm$ have the same spectrum. In particular this means that $U_\pm $ have the same sets of eigenvalues and the same characteristic polynomials.

\subsection{3.1 The asymptotic operators $L_\pm$}

The characteristic equation for  $L_\pm =Q_\pm -\lambda J$ takes the form:
\begin{equation} \label{curve}
\det(Q_\pm -\lambda J-z \openone)=-z(z^4-(\lambda^2+2a_2)z^2+a_4)=0,
\end{equation}
where
\begin{align*}
&a_2=p_{1,\pm} q_{1,\pm}+p_{2,\pm}q_{2,\pm}+p_{3,\pm}q_{3,\pm} = (\vec{p}_\pm \vec{q}_\pm), \qquad b_2^2 = (\vec{q}_\pm^T s_0 \vec{q}_\pm) (\vec{p}_\pm^T s_0 \vec{p}_\pm) ,\qquad
a_4=  a_2^2 -b_2^2 .
\end{align*}

The solution of the equation \eqref{curve} is given by:
\begin{equation*}
z=0,\quad z^2=\dfrac{\lambda^2+2a_2\pm\sqrt{D_1}}2, \qquad D_1=(\lambda^2+2a_2)^2-4a_4= \lambda^4 +4a_2\lambda^2 +4b_2^2.
\end{equation*}
 The roots of $D_1=0$ are $\lambda^2=-2a_2\pm2\sqrt{a_4}$.
Neglecting the root $z\ne 0$ equation \eqref{curve} takes the form:
\begin{equation*}
(\lambda z)^2=z^4-2a_2z^2+a_4=(z^2-a_2)^2+a_4-a_2^2.
\end{equation*}
The parametrization of the eigenvalues depends on the values of  the constants. In particular, if $0<a_4<a_2^2$, $a_2>0$ then we can introduce the notations:
\begin{gather*}
A_1^2 =a_2 - b_2, \qquad k^2 =\frac{ a_2 - b_2}{a_2 + b_2},\qquad z_{1,5} = \pm A_1\sn(u|k).
\end{gather*}
Then $\lambda$ can be expressed in terms of Jacobi elliptic functions, see  Appendix D. 
\begin{equation} \label{lambda.u}
\lambda=\pm \dfrac{A_1\cn(u|k)\dn(u|k)}{k\sn(u|k)}.
\end{equation}

Inserting \eqref{lambda.u} into \eqref{curve} and simplifying, we find for the roots:
\begin{equation}\label{eq:z24}
z_{2,4}=\mp\dfrac{A_1}{ k \sn(u|k)}.
\end{equation}
\begin{remark}\label{rem:2}
Note that if $\im z_1 >0$, then $\im z_2 >0$, which has to be taken into account when ordering the eigenvalues of $\mathcal{J}$.
\end{remark}

Thus to each value of the spectral parameter $\lambda$ given by \eqref{lambda.u} there correspond four different eigenvalues:
\begin{equation} \label{z.u}
z_{1,5}=\pm A_1\sn(u|k),\quad z_{2,4}=\mp\dfrac{A_1}{k \sn(u|k)}.
\end{equation}

Let us now calculate the eigenvectors of $L_\pm$. First we note that the eigenvector corresponding to the eigenvalue  $z_3=0$ takes the form:
\begin{equation*}
\Psi_3^t=(0,\; p_{1,\pm}q_{2,\pm}+p_{2,\pm}q_{3,\pm},\; p_{3,\pm}q_{3,\pm}-p_{1,\pm} q_{1,\pm}, \; -p_{2,\pm}q_{1,\pm}-p_{3,\pm}q_{2,\pm},\; 0)^T.
\end{equation*}
The  eigenvector corresponding to the eigenvalue $z_j$ is given by:
\begin{equation*}
\Psi_j= c_j\begin{pmatrix}(2q_{1,\pm}q_{3,\pm}-q_{2,\pm}^2)z_j\\ q_{3,\pm}z_j^2+q_{3,\pm} \lambda z_j+(p_{1,\pm}q_{1,\pm}-p_{2,\pm}q_{2,\pm}-p_{3,\pm}q_{3,\pm})q_{3,\pm}-p_{1,\pm}q_{2,\pm}^2\\
-q_{2,\pm}z_j^2-q_{2,\pm}\lambda z_j+(p_{1,\pm}q_{1,\pm}+p_{3,\pm}q_{3,\pm})q_{2,\pm} +2p_{2,\pm}q_{1,\pm}q_{3,\pm}\\
q_{1,\pm}z_j^2+q_{1,\pm}\lambda z_j+(p_{3,\pm}q_{3,\pm} -p_{2,\pm}q_{2,\pm}- p_{1,\pm}q_{1,\pm})q_{1,\pm}-p_{3,\pm}q_{2,\pm}^2\\
z_j^3+\lambda z_j^2-(p_{1,\pm}q_{1,\pm}+p_{2,\pm}q_{2,\pm}+p_{3,\pm}q_{3,\pm})z_j
\end{pmatrix}, \qquad L_\pm \Psi_j = z_j \Psi_j
\end{equation*}
where $c_j$ are norming constants and $\lambda$ and $z_j$ are as in \eqref{lambda.u} and \eqref{z.u}. We will need also the `left` eigenvectors of $L_\pm$:
\begin{equation}\label{eq:Psij}\begin{split}
\langle \Psi_j |( L_\pm +z_j)=0, \qquad \langle \Psi_j | = \Psi_j^T S_0.
\end{split}\end{equation}

Besides we need also to normalize correctly the eigenvectors so that the matrices:
\begin{equation}\label{eq:psi}\begin{split}
\Psi = \left(  \Psi(z_1), \Psi(z_2), \Psi(0), \Psi(-z_2), \Psi(-z_1) \right), \qquad \widehat{\Psi} = \left(\begin{array}{c} \langle \Psi(-z_1) | \\ \langle \Psi(-z_2) | \\ \langle \Psi(0) | \\ \langle \Psi(z_2) | \\ \langle \Psi(z_1) |  \end{array}\right)
\end{split}\end{equation}
take values in the group $SO(5)$, i.e. $\Psi S_0 \Psi^T S_0 =\openone_5$ and $\Psi \widehat{\Psi} =\openone_5 $. In addition we order the eigenvalues as follows: $z_1, z_2 , 0, z_4, z_5 $ so that  $z_5 =-z_1$, $z_4 =-z_2$. Introduce the notations:
\begin{equation}\label{eq:w0}\begin{split}
 w_{0,\pm} = \Psi_\pm, \qquad \widehat{w}_{0,\pm} = \widehat{\Psi}_\pm,
 \end{split}\end{equation}
where $\Psi_\pm$, $\widehat{\Psi}_\pm$ obtained from $\Psi$, $\widehat{\Psi}$ by replacing $q_j, p_j$ with $q_{j,\pm}, p_{j,\pm}$. This gives:
\begin{equation}\label{eq:w0pm}\begin{split}
w_{0,\pm}^{-1} (Q_\pm - \lambda J) w_{0,\pm} = -\mathcal{J}(u), \qquad \mathcal{J}(u) = \diag (z_1, z_2, 0, -z_2, -z_1).
\end{split}\end{equation}

\subsection{3.2 The Jost solutions }
Now we can proceed to define the Jost solutions of the Lax operator. To this end we first rewrite the Lax operator into more convenient form:
\begin{equation}\label{eq:Ltil}\begin{split}
 \tilde{L} \tilde{\psi} \equiv i \frac{\partial \tilde{\psi}}{ \partial x } + \tilde{Q}_\pm (x,t,u) \tilde{\psi}(x,t,u) - \mathcal{J}(u)  \tilde{\psi}(x,t,u)=0, \qquad \tilde{Q}_\pm (x,t,u) = w_{0,\pm}^{-1} (Q(x,t) - Q_\pm)  w_{0,\pm}.
\end{split}\end{equation}
Here we have transformed the asymptotic operators $L_\pm$ into diagonal form and also replaced the spectral parameter $\lambda$ by the uniformization variable $u$. The Jost solutions of the operator $\tilde{L}$ are defined by:
\begin{equation}\label{eq:psitil}\begin{split}
\lim_{x\to \infty} \tilde{\psi}(x,t,u) e^{i \mathcal{J}(u) x} = \openone, \qquad \lim_{x\to -\infty} \tilde{\phi}(x,t,u) e^{i \mathcal{J}(u) x} = \openone.
\end{split}\end{equation}

The Jost solutions satisfy the following integral equations:
\begin{equation}\label{eq:JoInt}\begin{split}
\tilde{\Psi}(x,u) & = \openone + i \int_{\infty}^{x} dy \; e^{-i \mathcal{J}(x-y)} \tilde{Q}_+(y,u) \tilde{\Psi}(y,u)  e^{i \mathcal{J}(x-y)} , \\
\tilde{\Phi}(x,u) & = \openone + i \int_{-\infty}^{x} dy \; e^{-i \mathcal{J}(x-y)} \tilde{Q}_-(y,u) \tilde{\Phi}(y,u)  e^{i \mathcal{J}(x-y)} ,
\end{split}\end{equation}
where $\tilde{\Psi}(x,u) = \tilde{\psi}(x,u) e^{i \mathcal{J}x}$, $\tilde{\Phi}(x,u) = \tilde{\phi}(x,u) e^{i \mathcal{J}x}$. In what follows we will need the integral equations (\ref{eq:JoInt}) in components:
\begin{equation}\label{eq:JoIntC}\begin{split}
\tilde{\Psi}_{jk}(x,u) & = \delta_{jk} + i \int_{\infty}^{x} dy \;  \left( \tilde{Q}_+(y,u) \tilde{\Psi}(y,u)\right)_{jk}  e^{-i (z_j - z_k)(x-y)} , \\
\tilde{\Phi}_{jk}(x,u) & = \delta_{jk} + i \int_{-\infty}^{x} dy \; \left( \tilde{Q}_-(y,u) \tilde{\Psi}(y,u)\right)_{jk}  e^{-i (z_j - z_k)(x-y)} ,
\end{split}\end{equation}
where we remind that $z_j(u)$ are the eigenvalues of $\mathcal{J}$. It will be easier to use equations (\ref{eq:JoIntC}) in discussing the analyticity properties of fundamental solutions of $\tilde{L}$.

Next we introduce the scattering matrix of the Lax operator and its inverse by:
\begin{equation}\label{eq:T}\begin{aligned}
T(t,u) &= \tilde{\psi}(x,u)^{-1} \tilde{\phi}(x,u), &\qquad T^{-1}(t,u) &= \tilde{\phi}(x,u)^{-1} \tilde{\psi}(x,u), \\
T(t,u) &= \left( \begin{array}{ccc} m_1^+ & -\vec{B}^-{}^T & c_1^- \\
\vec{b}^+ & {\bf T}_{22} & -s_0\vec{b}^- \\ c_1^+ & \vec{B}^+{}^Ts_0 & m_1^- \\
\end{array}\right), &\qquad
T^{-1}(t,u) &= \left( \begin{array}{ccc} m_1^- & -\vec{b}^-{}^T & c_1^- \\
-\vec{B}^+ & s_0{\bf T}_{22}s_0 & s_0\vec{B}^- \\ c_1^+ & -\vec{b}^+{}^Ts_0 & m_1^+ \\
\end{array}\right),
\end{aligned}\end{equation}

\begin{remark}\label{rem:1}
In what follows we will impose the following constraints on the potential $Q(x,t)$:
\begin{description}
  \item[i)] The potential $Q(x)$ tends fast enough to its limits, i.e.:
  \begin{equation}\label{eq:Q}\begin{split}
   \lim_{x\to \infty} x^p (Q(x,t) - Q_+) =0, \qquad \lim_{x\to -\infty} x^p (Q(x,t) - Q_-) =0, \qquad p =1,2,3, \dots;
  \end{split}\end{equation}

  \item[ii)] The potential $Q(x)$ is such that the operator $L$ has no discrete eigenvalues. Below we shall show that the discrete eigenvalues of $L$ are related to the zeroes of the principal minors of the scattering matrix $T(t,u)$.

  \item[iii)] Remember that $z_5 =- z_1$, $z_4 =-z_2$ and $z_3=0$.

\end{description}
\end{remark}

Let us now assume that $u$ is such that all eigenvalue $z_j(u)$ as well as $z_1(u) - z_2(u)$ are real. Then the exponential factors in the integral equations (\ref{eq:JoIntC}) will be oscillating and condition i) will ensure both the convergence of the integrals in the right hand sides of eqs. (\ref{eq:JoIntC}) and the existence of the Jost solutions. In fact all $u$ for which one can define the Jost solutions will determine the continuous spectrum of $L$.

\subsection{3.3 The FAS of $L$}

Let us now analyze the analyticity properties of the solutions to the integral equations (\ref{eq:JoIntC}) and let us assume that we have separated the regions $\mathcal{A}^+$ and $\mathcal{A}^-$ in the fundamental domain in which we have:
\begin{equation}\label{eq:ReA}\begin{split}
 \mathcal{A}^+ \quad \colon \quad \im z_1(u) > \im z_2(u) >0, \qquad  \mathcal{A}^- \quad \colon \quad -\im z_1(u) > -\im z_2(u) >0.
\end{split}\end{equation}

In the region $ \mathcal{A}^+$ (resp.  $\mathcal{A}^-$) all exponential factors $e^{-i(z_1 - z_k)(x-y)}$ in the first of the equations (\ref{eq:JoIntC}) will decrease exponentially. This is due to the condition (\ref{eq:ReA}) and to the fact that the integration is over the interval $[\infty, x]$ where $x-y <0$. These facts ensure that the first column of $\tilde{\Psi}(x,u)$ is analytic for $u \in \mathcal{A}$. Similarly we can consider the last column of $\tilde{\Phi}(x,u)$. In this case the exponential factors $e^{-i(z_5 - z_k)(x-y)}$ in the second of the equations (\ref{eq:JoIntC}) will also decrease exponentially taking into account iii) and the fact that  the integration is over the interval $(-\infty , x]$ in which $(x-y)$ is positive.

However, the other columns of $\tilde{\Psi}(x,u)$ and $\tilde{\Phi}(x,u)$ {\em will not } have analyticity properties for $u \in \mathcal{A}$, because some of the exponential factors will be increasing and the equations (\ref{eq:JoIntC}) will not have solutions for these values of $u$.

Nevertheless we are able to construct the fundamental analytic solution (FAS) of $\tilde{L}$ in the region $u\in \mathcal{A}^+$ using the idea of Shabat \cite{Sh*75, Sh*79, ZMNP}. Indeed, let us introduce $\chi_{A}^+$ and $\chi_{A'}^+$  as the solutions of the integral equations:
\begin{equation}\label{eq:JoIntC2}\begin{aligned}
\chi^+_{A;jk}(x,u) & = \delta_{jk} + i \int_{\infty}^{x} dy \;  \left( \tilde{Q}_+(y,u) \chi^+_{A}(y,u)\right)_{jk}  e^{-i (z_j - z_k)(x-y)} ,  &\qquad \mbox{for} \quad j< k\\
\chi^+_{A;jk}(x,u) & = \delta_{jk} + i \int_{-\infty}^{x} dy \; \left( \tilde{Q}_-(y,u) \chi^+_{A}(y,u)\right)_{jk}  e^{-i (z_j - z_k)(x-y)} , &\qquad \mbox{for} \quad j\geq k, \\
\chi^+_{A';jk}(x,u) & = \delta_{jk} + i \int_{\infty}^{x} dy \;  \left( \tilde{Q}_+(y,u) \chi^+_{A'}(y,u)\right)_{jk}  e^{-i (z_j - z_k)(x-y)} ,  &\qquad \mbox{for} \quad j\leq k\\
\chi^+_{A';jk}(x,u) & = \delta_{jk} + i \int_{-\infty}^{x} dy \; \left( \tilde{Q}_-(y,u) \chi^+_{A'}(y,u)\right)_{jk}  e^{-i (z_j - z_k)(x-y)} , &\qquad \mbox{for} \quad j> k,
\end{aligned}\end{equation}
Similarly we can write dow the `dual` equations for $\chi^-_{A;jk}(x,u)$ in the form:
\begin{equation}\label{eq:JoIntC3}\begin{aligned}
\chi^-_{A;jk}(x,u) & = \delta_{jk} + i \int_{\infty}^{x} dy \;  \left( \tilde{Q}_+(y,u) \chi^-_{A}(y,u)\right)_{jk}  e^{-i (z_j - z_k)(x-y)} ,  &\qquad \mbox{for} \quad j> k\\
\chi^-_{A;jk}(x,u) & = \delta_{jk} + i \int_{-\infty}^{x} dy \; \left( \tilde{Q}_-(y,u) \chi^-_{A}(y,u)\right)_{jk}  e^{-i (z_j - z_k)(x-y)} , &\qquad \mbox{for} \quad j\leq k, \\
\chi^-_{A';jk}(x,u) & = \delta_{jk} + i \int_{\infty}^{x} dy \;  \left( \tilde{Q}_+(y,u) \chi^-_{A'}(y,u)\right)_{jk}  e^{-i (z_j - z_k)(x-y)} ,  &\qquad \mbox{for} \quad j\geq k\\
\chi^-_{A';jk}(x,u) & = \delta_{jk} + i \int_{-\infty}^{x} dy \; \left( \tilde{Q}_-(y,u) \chi^-_{A'}(y,u)\right)_{jk}  e^{-i (z_j - z_k)(x-y)} , &\qquad \mbox{for} \quad j< k,
\end{aligned}\end{equation}

\begin{theorem}\label{thm:1}
 Let us assume that the potential $Q(x,t)$ satisfies the conditions in Remark \ref{rem:1}. Then the solutions $\chi_{A}^\pm(x,u)$ and $\chi_{A'}^\pm(x,u)$ are FAS of the operator $\tilde{L}$ for $u \in \mathcal{A}^\pm$ respectively.
\end{theorem}

\begin{proof}[Idea of the proof] It is easy to check that all exponential factors in the integral equations (\ref{eq:JoIntC2}) are either oscillating (for $j=k$) or exponentially decreasing. Therefore all integrals in the right hand sides of (\ref{eq:JoIntC2}) will be convergent. For $k=j$ this is ensured by condition i); for $k\neq j$ this is ensured by exponential factors.

\end{proof}

It is well known that any two fundamental solutions of the same operator are linearly related. In other words $\chi_{A}(y,u)$ is related to both Jost solutions. In order to find out the structure of the relevant factors we need to calculate the limits $\chi_{A}(y,u)$ for $x\to \pm \infty$. Skipping the details we find:
\begin{equation}\label{eq:chiApm}\begin{aligned}
 \chi_{A}^\pm(y,u) & = \tilde{\Phi}(x,u) S_A^\pm(t,u) , &\qquad  \chi_{A}^\pm(y,u) & = \tilde{\Psi}(x,u) T^\mp_A(t,u) D_A^\pm(u), \\
 \chi_{A'}^\pm(y,u) & = \tilde{\Phi}(x,u) S_A^\pm(t,u)\hat{D}_A^\pm(u) , &\qquad  \chi_{A'}^\pm(y,u) & = \tilde{\Psi}(x,u) T^\mp_A(t,u) ,
\end{aligned}\end{equation}
where by `hat` we have denoted the inverse matrix, i.e. $\hat{D}\equiv D^{-1}$.
Evaluating the limits of $\chi_{A}$ and $\chi^\pm_{A'}$ for $x\to \pm \infty$ we find that the factors $S_A^\pm$ (resp. $T_A^\pm$) must be upper-triangular (resp. lower-triangular) matrices with 1 on the diagonal, while $D^\pm_A$ must be a diagonal matrix. These matrices must be related to the scattering matrix by:
\begin{equation}\label{eq:Ttu}\begin{aligned}
 T(t,u) = T_A^- D_A^+ \hat{S}_A^+ = T_A^+ D_A^- \hat{S}_A^- .
\end{aligned}\end{equation}
As for the VBC case we have to apply the Gauss decomposition in order to express the Gauss factors $T_A^\pm$, $D_A^\pm$ and $S_A^\pm$ in terms of the matrix elements of $T(t,u)$, see  Appendix B. 

The analyticity regions of the FAS constructed above however cover only half of the fundamental rectangle. Along with them we have to consider FAS for which the ordering of $\im z_j$ is different, namely:
\begin{equation}\label{eq:ReB}\begin{split}
 \mathcal{B}^+ \quad \colon \quad \im z_2(u) > \im z_1(u) >0, \qquad  \mathcal{B}^- \quad \colon \quad -\im z_2(u) > -\im z_1(u) >0.
\end{split}\end{equation}
These cases can be treated quite analogously to the previous ones. Indeed, we can apply to the operator $L$ similarity transformation with $S_{\alpha_1}$ where:
\begin{equation}\label{eq:Salf}\begin{split}
 S_{\alpha_1} = \left(\begin{array}{ccccc} 0 & 1 & 0 & 0 & 0 \\  -1 & 0 & 0 & 0 & 0 \\  0 & 0 & \openone & 0 & 0 \\  0 & 0 & 0 & 0 & 1 \\ 0 & 0 & 0 & -1 & 0  \end{array}\right)
\end{split}\end{equation}
which corresponds to the Weyl reflection with respect to the root $e_1$. Applying this reflection to $\mathcal{J}$ we obtain $S_{\alpha_1} \mathcal{J} = \diag (z_2, z_1, 0 -z_1, -z_2)$. Then we can construct the FAS for the transformed operator $S_{\alpha_1}(L)$ whose scattering matrix will be
$T'(t,u)=S_{\alpha_1} (T(t,u))$. In order to get the corresponding FAS for our initial operator $L$ we need to apply the inverse transformation. But it is easy to check that $S_{\alpha_1}^2 = \openone$, so $S_{\alpha_1} = S_{\alpha_1}^{-1}$. The final result is that we have two more pairs of FAS  $\chi_{B}^\pm(y,u)$ and $\chi_{B'}^\pm(y,u)$ related to the Gauss decompositions of $T'(t,u)$ as follows:
\begin{equation}\label{eq:chiBpm}\begin{aligned}
 \chi_{B}^\pm(y,u) & = \tilde{\Phi}(x,u) S_B^\pm(t,u) , &\qquad  \chi_{B}^\pm(y,u) & = \tilde{\Psi}(x,u) T^\mp_B(t,u) D_B^\pm(u), \\
 \chi_{B'}^\pm(y,u) & = \tilde{\Phi}(x,u) S_B^\pm(t,u)\hat{D}_B^\pm(u) , &\qquad  \chi_{B'}^\pm(y,u) & = \tilde{\Psi}(x,u) T^\mp_B(t,u) ,
\end{aligned}\end{equation}
where $S_B^\pm(t,u) $, $D_B^\pm(u) $ and $T_B^\pm(t,u) $ are constructed from the scattering matrix $T(t,u)$ as follows: first we apply the Weyl reflection to $T(t,u)$ and find $T'(t,u)$. Then we construct the Gauss decompositions of $T'(t,u)$:
\begin{equation}\label{eq:Ttu'}\begin{aligned}
T'(t,u) = T^{\prime, -} D^{\prime,+} \hat{S}^{\prime,+} = T^{\prime, +} D^{\prime,-} \hat{S}^{\prime,-}.
\end{aligned}\end{equation}
Next we apply the inverse transformation to the Gauss factors and obtain the factors:
\begin{equation}\label{eq:TSB}\begin{split}
T_B^\pm(t,u) = S_{\alpha_1}(T^{\prime, \pm}(t,u)), \quad D_B^\pm(u) = S_{\alpha_1}(D^{\prime, \pm}(u)), \quad S_B^\pm(t,u) = S_{\alpha_1}(S^{\prime, \pm} (t,u)),
\end{split}\end{equation}
that relate the new FAS to the Jost solutions, see eq. (\ref{eq:chiBpm}).

These construction, unlike for the VBC require that we solve for generic  Gauss factors of the scattering matrix $T(t,u)$ or the transformed one  $T'(t,u)$. These formulae have been known for long time now, see the Appendix or \cite{ContMat} and the references therein.

\subsection{3.4 The $t$-dependence of the scattering matrix}

In order that the Lax pair remains compatible with the CBC we need to modify the $M$-operator.
We take in the form:
\begin{equation}\label{eq:UV2}\begin{aligned}
M\psi \equiv i \frac{\partial \psi}{ \partial t } + V^{(1)}(x,t,\lambda)\psi (x,t,\lambda) &= 0, \qquad V^{(1)}(x,t,\lambda) = i Q_{1,x}+ V_2(x,t) - V_{2,\pm} +\lambda Q(x,t) -\lambda^2 J,
\end{aligned}\end{equation}
where the additional term is $V_{2,\pm} = \lim_{x\to\pm\infty} V_2(x,t)$.

In order to derive the $t$-dependence of the scattering matrix $T(t,\lambda)$ we need to consider the asymptotic operators $M_\pm = i \frac{\partial }{ \partial t} + V_\pm (\lambda)$ which become:
\begin{equation}\label{eq:Mpm}\begin{split}
M_\pm \psi_\pm \equiv i \frac{\partial \psi_\pm }{ \partial t } + (\lambda Q_\pm - \lambda^2 J)\psi_\pm = 0,
\end{split}\end{equation}
As one could expect the asymptotic operators $M_\pm$ commute with $L_\pm$; so they are diagonalized by the same matrices $w_{0,\pm}$ introduced above. As a result we find that:
\begin{equation}\label{eq:Tt}\begin{split}
i \frac{\partial T}{ \partial t } - [ \lambda \mathcal{J}, T(\lambda,t)]=0.
\end{split}\end{equation}

\begin{figure}[htbp]
\begin{center}
\includegraphics[width=0.8\textwidth]{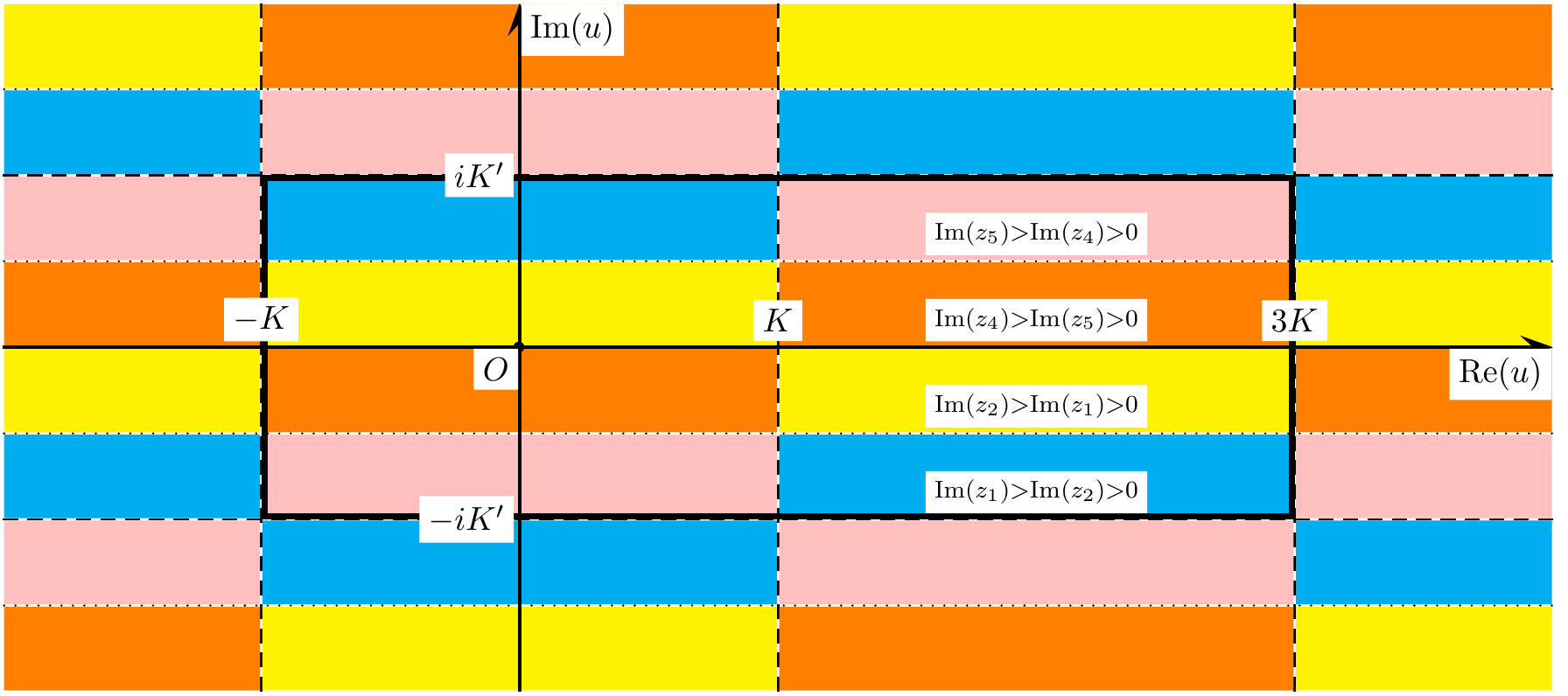}
\caption{The analytic regions on the plane: in cyan: $\im z_1 > \im z_2 >0 $; in pink: $ -\im z_1 > -\im z_2 >0 $; in orange: $-\im z_2 > -\im z_1 >0 $; in yellow: $\im z_2 > \im z_1 >0 $. The horizontal lines (parallel to the $\re u$ axis) are given by $\im u = N/2$, where $N= 0, \pm 1, \pm 2, $\dots  } \label{fig:sn0}
\end{center}
\end{figure}

It is easy to check that the Gauss factors of $T(t,u)$ satisfy:
\begin{equation}\label{eq:TAt}\begin{aligned}
i \frac{\partial T_A^\pm}{ \partial t } - [ \lambda \mathcal{J}, T_A^\pm(\lambda,t)] &=0, &\quad
i \frac{\partial S_A^\pm}{ \partial t } - [ \lambda \mathcal{J}, S_A^\pm(\lambda,t)] &=0, &\quad
i \frac{\partial D_A^\pm}{ \partial t } &=0, \\
i \frac{\partial T_B^\pm}{ \partial t } - [ \lambda \mathcal{J}, T_B^\pm(\lambda,t)] &=0, &\quad
i \frac{\partial S_B^\pm}{ \partial t } - [ \lambda \mathcal{J}, S_B^\pm(\lambda,t)]= &0, &\quad
i \frac{\partial D_B^\pm}{ \partial t } &=0.
\end{aligned}\end{equation}
In other words $D_A^\pm (\lambda)$ and $D_B^\pm (\lambda)$ are generating functionals of the integrals of motion for the KSM.

\subsection{3.5 The inverse scattering problem for CBC}

Again we will reduce the inverse scattering problem to a RHP. The important difference between VBC and CBC case is that we have to formulate the RHP on the Riemannian surface of genus 1. Equivalently we can formulate it on the fundamental rectangle on the complex $\lambda$-plane, see Figure \ref{fig:sn0}. Now we have 4 domains of analyticity painted by different colors $\mathcal{A}^\pm$ and $\mathcal{B}^\pm$. The lines that separate these domains constitute the continuous spectrum of $L$ with CBC.
The corresponding RHP requires the sewing functions that  relate the FAS on both sides of those lines, see Figure \ref{fig:2}:

\begin{figure}
\begin{center}
\includegraphics[width=0.5\textwidth]{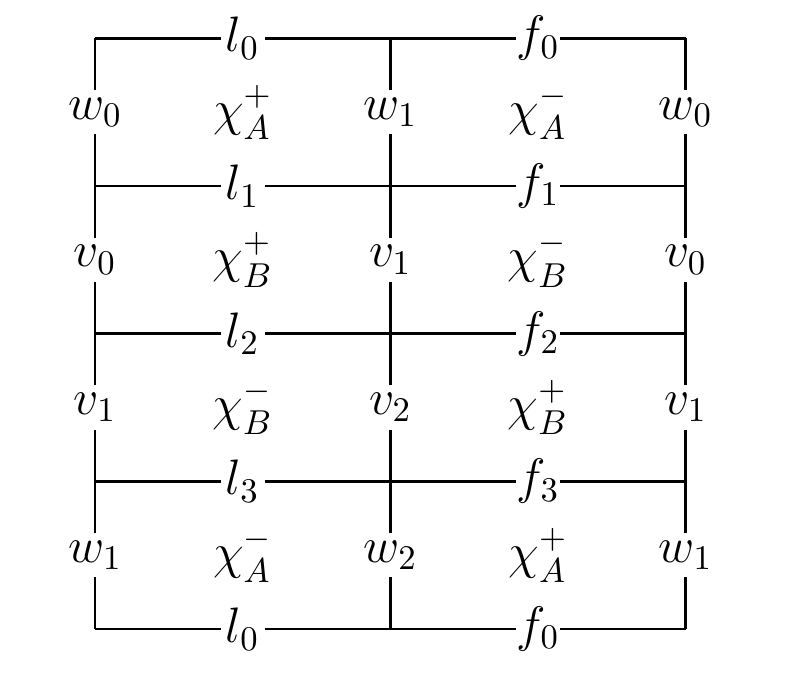}
\caption{The analytic regions on fundamental rectangular. By $f_k$, $l_k$, $k=0,1,2,3$ and $v_j$, $w_j$, $j=0,1,2$ we have denoted the segments separating the analyticity regions of the FAS $\chi^\pm_A$ and $\chi^\pm_B$.  } \label{fig:2}
\end{center}
\end{figure}

\begin{equation}\label{eq:RHA}\begin{aligned}
 \chi^+_A(x,t,u) &= \chi^-_A(x,t,u) G_A(t,u), & \quad  G_A &= \hat{S}_A^- S_A^+(t,u),  &\qquad u & \in l_0 \cup f_0; &\quad u &\in w_0\cup w_1; \\
 \chi^{\prime +}_A(x,t,u) &= \chi^{\prime ,-}_A(x,t,u) G'_A(t,u), &\quad G_A' &= \hat{T}_A^+ T_A^-(t,u),  &\qquad u &\in l_0 \cup f_0; &\quad u &\in w_0\cup w_1; \\
\chi^+_B(x,t,u) &= \chi^-_B(x,t,u) G_B(t,u), & \quad  G_B &= \hat{S}_B^- S_B^+(t,u), &\qquad u &\in l_2 \cup f_2; &\quad u &\in v_0\cup v_1; \\
 \chi^{\prime +}_B(x,t,u) &= \chi^{\prime ,-}_B(x,t,u) G'_B(t,u), &\quad G_B' &= \hat{T}_B^+ T_B^-(t,u). &\qquad u &\in l_2 \cup f_2; &\quad u &\in v_0\cup v_1,
\end{aligned}\end{equation}
and
\begin{equation}\label{eq:RHB}\begin{aligned}
 \chi^+_A(x,t,u) &= \chi^+_B(x,t,u) G^+_{AB}(t,u), & \quad  G_A &= \hat{S}_B^+ S_A^+(t,u),  &\qquad u &\in l_1 \cup f_3;\\
 \chi^-_A(x,t,u) &= \chi^-_B(x,t,u) G^-_{AB}(t,u), & \quad  G_A &= \hat{S}_B^- S_A^-(t,u),  &\qquad u &\in f_1 \cup f_3.
\end{aligned}\end{equation}
where the segments $l_j, f_j$, $j=0,1,2,3$ and $w_k, v_k$, $k=0,1,2$ are shown on Figure \ref{fig:2}. The $t$-dependence of the sewing functions $ G_A(t,u)$, $ G_A'(t,u)$, $ G_B(t,u)$, $ G_B'(t,u)$, $ G^+_{AB}(t,u)$ and $ G^-_{AB}(t,u)$ follows from equations (\ref{eq:TAt}).

The fact that eigenvalues of the asymptotic operators are expressed in terms of Jacobi elliptic functions has as a consequence that the RHP is formulated on the genus 1 Riemannian surface, or equivalently, on the fundamental rectangular with identified opposite sides.


Of course these conditions are doubly periodic as the Jacobi elliptic functions. They may provide a basis for applying the dressing method for deriving the soliton solutions for NBC.

Note that these results are valid for $(2n+1)\times (2n+1)$ Lax operators whose potentials are compatible with the symmetric spaces $SO(2n+1)/S(O(2n-1) \times O(2))$. Indeed, it is easy to check that the characteristic polynomials of the corresponding  asymptotic operators $L_\pm$ take the form:
\begin{equation} \label{curveN}
\det(Q_\pm -\lambda J-z \openone)=-z^{2n-3}(z^4-(\lambda^2+2a_2)z^2+a_4)=0,
\end{equation}
where $a_2$ and $a_4$ are given by:
\begin{equation}\label{eq:CPla}\begin{split}
a_2= (\vec{p},\vec{q}), \qquad b_2^2 = (\vec{q}^T s_0 \vec{q}) (\vec{p}^T s_0 \vec{p}) ,\qquad
a_4=  a_2^2 -b_2^2 .
\end{split}\end{equation}

It will be interesting to consider the particular cases, when the Jacobi elliptic functions turn into trigonometric functions or hyperbolic trigonometric functions. These and other related questions will be considered in next publications.

\section{4. Discussion and conclusion}

We demonstrated that Shabat's method for constructing FAS can be generalized to Lax operators related to BD.I-type symmetric spaces and responsible for the integrability of KSM, which describes spin-1 BEC. The generalized  KSM with $2n-1$ components, $n>2$ describe higher spin BEC. For example 5-component KSM describes spin-2 BEC. In fact the treatment of generalized KSM does not change very much the results described above. In fact, from eq. (\ref{curveN}) we find that the set of non-vanishing eigenvalues of the operators $L_\pm$ are the same for any values of $n\geq 2$. What changes is the number of vanishing eigenvalues which is $2n-3$. As a result the factors in the  Gauss decompositions (\ref{eq:Ttu}) will be of generalized form containing central block of dimension $(2n-3) \times (2n-3)$, similarly to the form for the vanishing boundary conditions.

Solving the generalized RHP formulated by the equations (\ref{eq:RHA}) and (\ref{eq:RHB}) does not seem an easy task.  Again one may try to apply the Zakharov-Shabat dressing method \cite{ZaSha1, ZaSha12, Za*Mi}, see also \cite{I04}. The corresponding projectors obviously would be rather complicated functions of Jacobi elliptic functions.

We started with potentials satisfying generic CBC and demonstrated that the spectrum of the corresponding Lax operators is on genus-1 Riemannian surface. We can plot it on the complex $\lambda$-plane by: a) first plotting it on the fundamental rectangle (see Figure \ref{fig:sn0} and b) extending it to the whole complex plane using the double periodicity. We can also outline the spectrum of the Lax operators related to the symmetric spaces $SO(2n+1)/S(O(2n-3) \times O(4))$. The corresponding MNLS will be matrix generalizations of the KSM. The resulting spectrum will be on Riemannian surfaces of higher genus.

On the other hand we could consider quasi-periodic boundary conditions, see \cite{DMNe, Dub81e, Dub85e, BBEIM, MatSmi, MatSmi2, AS01, 152, 152b, 152c} and the numerous references therein. There the potential $Q(x,t)$ is a multi-quasi-periodic function related to a higher genus curve while the continuous spectrum of the Lax operator fills up a set of segments on the real axis in $\mathbb{C}$. It will be instructive to interrelate these two important classes of boundary conditions.

It may become easier if we consider limiting cases in which Jacobi elliptic functions become trigonometric, or hyper trigonometric functions. Before doing such simplifications however one must be careful. Indeed, the spin-1 BEC reduces to the KSM model after several assumptions for the physical constants (scattering lengths) of the BEC. So there comes up the question: will such simplifications as the one mentioned above involve additional constraints on the experimental parameters of the BEC.

There are additional open problems which could be analyzed. These concern MNLS related to some of the other classes of symmetric spaces (e.g. of C.I or D.III types) with additional Mikhailov-type reductions \cite{Mikh, Pliska12, VG-RomJP, Pliska16}.  The extension of these results to the corresponding CBC cases, along with the treatment of (possibly multiple) branching points  will be another nontrivial problem.

The last issue that we would like to mention here concerns the generalized Fourier transforms for systems with CBC. For the scalar case the completeness of the squared solutions  was derived by Konotop and Vekslerchik \cite{KonVe}. The difficulties here are concerned with the nonlinearity of the phase space and the proper use of the recursion operators.

\begin{acknowledgments}
One of us (VSG) acknowledges the support from Bulgarian Science Foundation under contract NSF   KP-06N42-2. This work was supported by the Ministry of Science and Higher Education of the Russian Federation (grant agreement No. FSRF-2020-0004).

\end{acknowledgments}

\appendix

\section{A. On BD.I symmetric spaces}\label{ap:BDI}

We will start with the basic properties of the simple Lie algebras $so(2n+1)$ \cite{Helg}. The root systems of these algebras is known as the $B_n$-series in Cartan classification. The sets of simple and positive roots $\Delta^+$ are as follows:
\begin{equation}\label{eq:Bn0}\begin{split}
\alpha_j = e_j - e_{j+1}, \quad j=1,\dots n-1, \quad \alpha_n=e_n, \qquad
 \Delta^+ \equiv \{ e_j \pm e_k, \quad 1 \leq j < k \leq n; \quad e_k \}.
\end{split}\end{equation}
For the sake of convenience we have modified the definition of orthogonality. We will say that $X\in so(2n+1)$ if $X + S_0 X^T S_0 =0$, where the matrix $S_0$ is given by:
\begin{equation}\label{eq:S0}\begin{split}
 S_0 = \sum_{k=1}^{2n+1} (-1)^k E_{k,2r-k} = \left(\begin{array}{ccc} 0 & 0 & 1 \\ 0 & -s_0 & 0 \\ 1 & 0 & 0   \end{array}\right)
\end{split}\end{equation}
where $s_0$ and $E_{km}$ were defined in (\ref{eq:KSm}). With this definition the Cartan subalgebra of $so(2n+1)$ is represented by diagonal matrices and the Cartan-Weyl basis takes the form:
\begin{equation}\label{eq:SOn}\begin{aligned}
 H_{j} &= E_{jj} - E_{\bar{j},\bar{j}}, \quad E_{e_j-e_k} = E_{jk} - (-1)^{k+j} E_{\bar{k}, \bar{j}}, \quad E_{e_j+e_k} = E_{j\bar{k}} - (-1)^{k+j} E_{k, \bar{j}}, \quad E_{e_j} = E_{j,n+1} -(-1)^{n+j} E_{n+1, \bar{j}},
\end{aligned}\end{equation}
where $\bar{k} = 2n+1-k$.

It is well known that the KSM is related to the symmetric space BD.I which is isomorphic to $SO(2n+1)/S(O(2n-1) \times O(2))$; for details about the structure of this and other symmetric spaces see \cite{Helg}.

\section{B. Gauss decompositions }\label{ap:Gauss}

The Gauss decompositions mentioned above have natural group-theoretical
interpretation and can be generalized to any semi-simple Lie algebra. It
is well known that if given group element allows Gauss decompositions then
its factors are uniquely determined. Below we write down the explicit
expressions for the matrix elements of $T^\pm(\lambda ) $, $\hat{S}^\pm(\lambda) $, $D^\pm(\lambda  ) $ through the matrix  elements of $T(\lambda ) $:
\begin{equation}\label{eq:Gaus0}\begin{aligned}
T^-_{pj}(\lambda ) &= {1 \over m_{j}^+(\lambda ) } \left\{
\begin{array}{ccccc} 1, & 2, & \dots , & j-1, & p \\ 1, & 2, & \dots , &
j-1, & j \end{array}\right\}_{T(\lambda )}^{(j)}, &\quad
\hat{S}^+_{jp}(\lambda ) &= {1  \over m_{j}^+(\lambda )}
\left\{\begin{array}{ccccc} 1, & 2, & \dots ,& j-1, & j \\
1, & 2, & \dots , & j-1, & p \end{array}\right\}_{T(\lambda )}^{(j)},\\
T^+_{pj}(\lambda ) &= {1 \over m_{n-j+1}^-(\lambda ) } \left\{
\begin{array}{cccc} p, & j+1, & \dots ,  & n \\ j, & j+1, & \dots ,
 & n \end{array}\right\}_{T(\lambda )}^{(n-j+1)}, &\quad
\hat{S}^-_{jp}(\lambda ) &= {1 \over m_{n-j+1}^-(\lambda ) }
\left\{ \begin{array}{cccc} j, & j+1, & \dots , & n \\
p, & j+1, & \dots ,  & n \end{array}\right\}_{T(\lambda)}^{(n-j+1)},\\
 D^+(\lambda ) &= \diag (D_1^+, D_2^+, \dots, D_n^+), &\quad
D^-(\lambda ) &= \diag (D_1^-, D_2^-, \dots, D_n^-),
\end{aligned}\end{equation}
where $m_j^+(\lambda )$ (resp. $m_j^-(\lambda )$) are the principle upper (resp. principle lower) minors of $T(\lambda)$ of order $j$ and
\begin{equation}\label{eq:10.7}
D^+_{j}(\lambda ) = {m_j^+(\lambda )  \over m^+_{j-1}(\lambda ) }, \quad
D^-_{j}(\lambda ) = {m_{n-j+1}^-(\lambda )  \over m^-_{n-j}(\lambda )},
\quad
\left\{ \begin{array}{cccc} i_1, & i_2, & \dots ,&i_k, \\
j_1, & j_2, & \dots , & j_k \end{array}\right\}_{T(\lambda)}^{(k)} =
\det \left| \begin{array}{cccc} T_{i_ij_1} & T_{i_1j_2} & \dots T_{i_1j_k}
\\ T_{i_2j_1} & T_{i_2j_2} & \dots T_{i_2j_k} \\ \vdots & \vdots & \ddots
& \vdots \\ T_{i_kj_1} & T_{i_kj_2} & \dots T_{i_kj_k} \end{array} \right|
\end{equation}
is the minor of order $k $ of $T(\lambda ) $ formed by the rows $i_1 $, $
i_2 $, \dots , $i_k $ and the columns $j_1 $, $j_2 $, \dots, $j_k $.

{}From the formulae above we arrive to the following
\begin{corrollary}\label{cor:1.1}
In order that the group element $T(\lambda ) \in SL(n,\bbbc) $
allows the first  (resp. the second)  Gauss decomposition
of $T(t,\lambda)$ is necessary and sufficient that all upper- (resp.
lower-) principle minors $m_k^+(\lambda ) $ (resp. $m_k^-(\lambda
) $) are not vanishing.
\end{corrollary}

These formulae hold true also if we need to construct the Gauss
decomposition of an element of the orthogonal $SO(n) $ group. Here we
just note that if $T(\lambda )\in SO(n) $ then
\begin{equation}\label{eq:def-G}
S_0 (T(\lambda ))^T S_0^{-1} = T^{-1}(\lambda ).
\end{equation}
One can check that if $T(\lambda ) $ satisfies (\ref{eq:def-G})
then each of the factors $T^\pm(\lambda ) $, $S^\pm(\lambda ) $
and $D^\pm(\lambda ) $ also satisfy (\ref{eq:def-G}) and thus
belong to the same group $\fr{G} $. In addition we have the
following interrelations between the principal minors of
$T(\lambda ) $:
\begin{eqnarray}\label{eq:m_k}
m_j^\pm(\lambda ) = m_{n-j}^\pm(\lambda ), \qquad \text{for} \quad SO(n),
\qquad \mbox{and} \qquad
m_j^\pm(\lambda ) = m_{n-j}^\pm(\lambda ), \qquad \text{for} \quad SP(n).
\end{eqnarray}

\section{C. Gauss decompositions for BD.I symmetric spaces}\label{ap:BDGauss}

The Gauss decompositions compatible with a given symmetric space and used in (\ref{eq:TGauss}) are generalizations of the well known ones. Here we will briefly outline how one can derive the equations (\ref{eq:FAS_J}). First we remember that $T(t,\lambda) \in SO(2n+1)$ and therefore its inverse $\hat{T}(t,\lambda) = S_0 T^T(t,\lambda)S_0$, see eq. (\ref{eq:25.1}). Let us now use the Cartan-Weyl basis to evaluate $S^\pm_J(t,\lambda)$ and $S^\pm_J(t,\lambda)$ in (\ref{eq:SJpm}):
\begin{equation}\label{eq:GausBDI}\begin{aligned}
S^+_J(t,\lambda) & = \left(\begin{array}{ccc} 1 & \vec{\tau}^{+,T} & c_1^+ \\ 0 & \openone & s_0\vec{\tau}^+ \\ 0 & 0 & 1 \end{array}\right), &\qquad T^+_J(t,\lambda) & = \left(\begin{array}{ccc} 1 & \vec{\rho}^{-,T} & \tilde{c}_1^+ \\ 0 & \openone &  s_0\vec{\rho}^- \\ 0 & 0 & 1 \end{array}\right), \\
S^-_J(t,\lambda) & = \left(\begin{array}{ccc} 1 & 0 & 0 \\ \vec{\tau}^{-}  & \openone & 0 \\ c_1^- & \vec{\tau}^{-,T} s_0 & 1 \end{array}\right), & \qquad T^-_J(t,\lambda) & = \left(\begin{array}{ccc} 1 & 0 & 0 \\ \vec{\rho}^{+}  & \openone & 0 \\ \tilde{c}_1^- & \vec{\rho}^{+,T} s_0 & 1 \end{array}\right),
\end{aligned}\end{equation}
where $c_1^\pm = \frac{1}{2} (\vec{\tau}^{\pm ,T} s_0 \vec{\tau}^\pm)$, $\tilde{c}_1^\pm = \frac{1}{2} (\vec{\rho}^{\mp,T} s_0 \vec{\rho}^\mp)$. The first relation in (\ref{eq:TGauss}) is equivalent to $T S_J^+ = T_J^- D_J^+$. From the first column we immediately get $\vec{\rho}^+ = \vec{b}^+/m_1^+$. Considering the last column of $T S_J^- = T_J^+ D_J^-$ we obtain $\vec{\rho}^- = -\vec{b}^-/m_1^-$. The expressions for $\vec{\tau}^\pm$ are obtained similarly considering the relations $\hat{T} T_J^+ = S_J^- \hat{D}_J^-$ and  $\hat{T} T_J^- = S_J^+ \hat{D}_J^+$.

\section{D. On Jacobi elliptic functions}\label{ap:Jacobi}

Here we will list some of the basic properties of Jacobi elliptic functions, for more details see  \cite{SFe, Akh}. We start with:
\begin{equation*}
\sn^2(u|k)+\cn^2(u|k)=1,\quad \dn^2(u|k)+k^2\sn^2(u|k)=1.
\end{equation*}
Their derivatives are given by:
\begin{equation*}
\begin{aligned}
&\dfrac{d}{du}\sn(u|k)=\cn(u|k)\dn(u|k),\quad
\dfrac{d}{du}\cn(u|k)=-\sn(u|k)\dn(u|k),\quad
\dfrac{d}{du}\dn(u|k)=-k^2\sn(u|k)\cn(u|k).
\end{aligned}
\end{equation*}
One can evaluate them for special values of the argument $u$ as follows:
\begin{align*}
&\sn(0|k)=0,\quad &&\cn(0|k)=1,\quad &&\dn(0|k)=1,\\
&\sn(K/2|k)=\dfrac1{\sqrt{1+k'}},\quad &&\cn(K/2|k)=\sqrt{\dfrac{k'}{1+k'}},\quad &&\dn(K/2|k)=\sqrt{k'},\\
&\sn(K|k)=1,\quad &&\cn(K|k)=0,\quad &&\dn(K|k)=k',\\
&\sn(iK'/2|k)=i/\sqrt{k},\quad &&\cn(iK'/2|k)=\sqrt{\dfrac{1+k}k},\quad &&\dn(iK'/2|k)=\sqrt{1+k},\\
&\sn(iK'|k)=\infty,\quad &&\cn(iK'|k)=\infty,\quad &&\dn(iK'|k)=\infty.
\end{align*}
where $K$ and $K'$ are the elliptic integrals of first kind:
\begin{equation*}
K=\int_0^1\dfrac{dt}{\sqrt{(1-t^2)(1-k^2t^2)}},\qquad K'=\int_0^1\dfrac{dt}{\sqrt{(1-t^2)(1-(k')^2t^2)}},
\end{equation*}
and $k^2+(k')^2=1$. In addition $ \sn(-u|k)=-\sn(u|k)$, $\cn(-u|k)=\cn(u|k)$, $\dn(-u|k)=\dn(u|k)$ and
\begin{align*}
&\sn(u\pm 2K|k)=-\sn(u|k),\quad && \sn(u\pm 2iK'|k)=\sn(u|k),\quad &&\cn(u\pm 2K|k)=-\cn(u|k),\\
& \cn(u\pm 2iK'|k)=-\cn(u|k),\quad &&\dn(u\pm 2K|k)=\dn(u|k),\quad && \dn(u\pm 2iK'|k)=-\dn(u|k);\\
&\sn(u+ K|k)=\dfrac{\cn(u|k)}{\dn(u|k)},\quad && \sn(u+iK'|k)=\dfrac{1}{k\sn(u|k)},\quad
&&\cn(u+ K|k)=- k'\dfrac{\sn(u|k)}{\dn(u|k)},\\
& \cn(u+iK'|k)=-i\dfrac{\dn(u|k)}{k\sn(u|k)},\quad &&\dn(u+K|k)=\dfrac{k'}{\dn(u|k)},\quad && \dn(u+iK'|k)=-i\dfrac{\cn(u|k)}{\sn(u|k)}.
\end{align*}
In order to evaluate the imaginary parts of the eigenvalues $z_k$ we will need also:
\begin{align*}
&\sn(x+iy|k)=\dfrac{s\cdot d_1+ic\cdot d\cdot s_1\cdot c_1}{c_1^2+k^2\cdot s^2\cdot s_1^2},\quad \cn(x+iy|k)=\dfrac{c\cdot c_1-is\cdot d\cdot s_1\cdot d_1}{c_1^2+k^2\cdot s^2\cdot s_1^2},\quad
\dn(x+iy|k)=\dfrac{d\cdot c_1\cdot d_1-ik^2\cdot s\cdot c\cdot s_1}{c_1^2+k^2\cdot s^2\cdot s_1^2},
\end{align*}
where $s=\sn(x|k)$,  $c=\cn(x|k)$, $d=\dn(x|k)$, and $s_1=\sn(y|k')$, $c_1=\cn(y|k')$ and  $d_1=\dn(y|k')$.



\end{document}